\documentclass[12]{amsart}
\usepackage{amsmath, amssymb, verbatim, epsfig, graphics}

\newtheorem{theorem}{Theorem}

\newtheorem{proposition}{Proposition}
\newtheorem{lemma}{Lemma}
\newtheorem{definition}{Definition}

\newcommand{\R}{\mathbb R}

\newcommand{\F}{\mathcal F}

\begin{document}

\title[Gradient Flow of SRB entropy]{Gradient Flow of the Sinai-Ruelle-Bowen Entropy }

\author[Miaohua Jiang]{Miaohua Jiang \\   \\   Department of Mathematics \\
Wake Forest University\\
Winston Salem, NC 27109, USA      \\ \\
jiangm@wfu.edu
        }



\begin{abstract}We study both the local and global existence of a gradient flow of the Sinai-Ruelle-Bowen entropy functional on a Hilbert manifold of expanding maps of a circle equipped with a
 Sobolev norm in the tangent space of the manifold. We show that, under a slightly modified metric, starting from
any initial value, every trajectory of the gradient flow converges to the map with a constant
expanding rate where the entropy attains the maximal value.
In a simple case, we obtain an explicit formula for the flow's ordinary differential equation representation.
This gradient flow has close connection to a nonlinear partial differential equation:
a gradient-dependent diffusion equation.
\end{abstract}

\maketitle

\large
\baselineskip 16 pt

\centerline{\bf\Large }

\section{Introduction }

Let $f_0$ be a $C^r , r \ge 3$ transitive Anosov map
or an  expanding map on a compact Riemannian manifold $M$.
 Let $U(f_0)$ be the connected open component of all $C^r$ maps in either family topologically conjugate to
 $f_0$. It is well-known that there exists a unique Sinai-Ruelle-Bowen (SRB) measure $\rho_f$ for every map
  $f \in U(f_0)$ \cite{Young} and the
  entropy  of the map with respect to the SRB measure $\rho_f$ is given by the formula
   $H(f)= \int_M \ln J^uf \ d \rho_f$ \cite{KH,Mane}, where $J^uf$ is the Jacobian of $f$ along the
   unstable subspace. We  call it the SRB entropy of $f$. This entropy
  is a Fr\'echet differentiable functional in $f$ with respect to the $C^r$ topology on   $U(f_0)$ \cite{Belse, J12, Ru97a}.

Recently, it has been proven that
   in low dimensional cases,  this entropy functional $H(f)$ does not have non-trivial critical points:
   if the value of the functional is less than its global maximum, then, there exists a
    direction in which the value of the functional increases when the map is perturbed.
    The global maximum of the entropy functional is only reached at maps that are
    smoothly conjugate to a linear map  \cite{J21, J22,SVV}.
    In view of Gallavotti-Cohen Chaotic Hypothesis \cite{G96,G06,GC, GR}, this property can be regarded as the
     realization of the Second Law of Thermodynamics in  mathematical models of chaotic dynamical systems.
       If we hypothesize that each map
     $f$ in $U(f_0)$ represents a possible state of a thermodynamic system,  the SRB entropy of
     $f$ corresponds to the Boltzmann entropy, and a thermodynamic system not at its
     equilibrium  evolves in the direction along which the entropy increases the fastest,
      we  encounter naturally the following two related questions:
One,
\emph{does the fastest increasing direction, or the direction of the gradient vector, of the SRB entropy exist}?
Two, \emph{ if the gradient vector field exists,  does the SRB entropy induce a gradient flow}?
The answers to these two questions are not
automatic since the SRB entropy functional is defined on a Banach manifold $U(f_0)$ with an infinite
 dimensional tangent space.
In this article, we  give affirmative answers to these two questions in a simple case where maps
are just expanding maps on a circle.
Note that in this simple case $U(f_0)$ consists of all $C^r$ expanding maps with the same degree $n\ge 2$.

In next section, we first extend the family $U(f_0)$ (still denoted by $U(f_0)$) to include
 all expanding maps whose $r$th ($r \ge 2$) derivative is $L^2$ and show that $U(f_0)$ is then equipped with
 a natural Hilbert manifold structure with a Sobolev norm in the tangent space. Under this Hilbert manifold
 structure, we show that the SRB entropy functional $H(f)$ remains Fr\'echet differentiable and thus the functional
 gives a gradient vector field over $U(f_0)$.  In Section \ref{sec.local}, we show the gradient vector
  field  is at least  Lipschitz continuous, which guarantees the existence of
   a local gradient flow. We note that a differentiable function's gradient is metric
    dependent. Under a slightly modified metric, we show in Section \ref{sec.global}
    the global existence of a gradient flow of the SRB entropy functional and the
  convergence of the flow to the linear expanding map as time approaches infinity.  In the last section,
   via harmonic  analysis, we obtain an ordinary differential equation representation of the
  gradient flow over the Hilbert manifold equipped with the Sobolev norm
  and give an example of a typical orbit using numerical approximation.
  The gradient flow also leads to a gradient-dependent diffusion equation on the circle.

\bigskip

\section{ Hilbert manifold structure on the family of circle expanding maps }\label{sec.local}

The SRB entropy is a differentiable functional in the space of
 $C^r , (r \ge 3)$  expanding maps on the unit circle. But when we consider the gradient of
 the entropy functional, the $C^r$ norm may not be the most convenient one. So, a
  Sobolev norm becomes a more natural choice instead of the $C^r$-norm.
  We first give the definition of a gradient vector for any Gateaux differentiable functional
  on the Hilbert manifold $\mathcal M$ with a tangent space $T_p\mathcal M=T$
   at each point $p \in \mathcal M$ and a Hilbert metric (inner product) $<\cdot >_\mathcal M$.

  \begin{definition}\label{df.gradient}
  A vector $V \in T_p\mathcal M$ is called a gradient vector of a Gateaux differentiable
   functional $H$ if $H'$s Gateaux derivative defines a bounded linear functional on $T_p\mathcal M$ with its Riesz representation equal to $V$.
  \end{definition}

We now describe a Hilbert manifold structure on a family of circle expanding maps.

\subsection{\bf Hilbert manifold of expanding maps on the circle}

First of all,  by considering its lift, we identify every continuous map $f$ on
$S^1=\{ e^{ i 2\pi x},  x \in [0, 1)\}$ with a function $\tilde f$ defined on the
real line satisfying the conditions $ \tilde f(0) \in [0, 1)$ and
$\tilde f(x+1) = \tilde f (x) +n$, where $n$ is the degree of the map.
We only consider orientation preserving maps. The case for orientation
reversing maps is essentially the same. Since   expanding maps are defined
on the circle,  there must be a fixed point and  we may assume $0$ is a fixed point:
$f(0)=0=\tilde f(0)$ and $\tilde f(1) =n$. Thus, each $C^r, r \ge 1$  map $f$ on the circle is
 identified with its lift $\tilde f$ on $\R$,  the universal covering of $S^1$,
   with the following properties:
   $ \tilde  f(0)=0,  \tilde   f(1) = n, \tilde   f'(x) > 1,$ and $  \tilde  f^{(k)}(0)= \tilde  f^{(k)}(1) ,  k=1,2,\cdots, r.$
    We now define a family of $C^{r-1}$ expanding maps $F_r, r\ge 2$ via properties of their lift maps
     where we consider the gradient flow of the entropy functional:
$$ F_r =\{f:    f  \in C^{r-1} (S^1),  \tilde f(0)=0,  \tilde f(1) = n,  \tilde f'(x) > 1, x \in (0,1), $$
 $$\qquad  \tilde  f^{(k)}(0)=\tilde  f^{(k)}(1) ,  k=1,2,\cdots, r-1,  \tilde f^{(r)}    \in L^2[0,1] \}.$$
 The family $F_r$ is slightly larger than the $C^r$ family of expanding maps since
 we only require the $r$th   derivative $f^{(r)}$ to be $L^2$, instead of being continuous.
$F_r$  is a Hilbert manifold modeled on Sobolev space $H^r$.  For any given map $ f\in F_r$,
  its open neighborhood
is identified with an open neighborhood of the origin of the following Hilbert space
$$ \Phi_r =\{  \phi(x) \in C^{r-1} [0,1]: \phi^{(r)} \in L^2[0,1],\phi(0)=0,
  \phi^{(k)}(0)=\phi^{(k)}(1), 0\le k \le r-1.\}$$
equipped with the Sobolev norm
\[  \| \phi(x) \|^2_{H^r} =  \sum_{k=0}^r  \int_0^1 [\phi^{(k)}]^2 (x) dx  .\]

Notice that $\Phi_r$ can be identified with a Sobolev sequence space
$$ \{\phi(x)= \sum_{n=0}^\infty  a_n \cos 2 \pi n x + \sum_{n=1}^\infty b_n \sin 2 \pi n x: \
   \sum_{n=1}^\infty  n^{2r} ( a_n^2 + b_n^2 ) < \infty,  \phi(0)=0\}$$  equipped with the
    corresponding Sobolev norm
$$\| \phi(x) \|^2_{H^r}= a_0^2 + \sum_{k=0}^r  \sum_{n=1}^\infty  n^{2k} ( a_n^2 + b_n^2 ).$$
where $\{a_n\}_{n=0}^\infty$ and $\{b_n\}_{n=1}^\infty$ are Fourier coefficients of $\phi$.

When $r \ge 3$, each expanding map $ f \in F_r$
possesses a unique invariant probability measure absolutely continuous with respect to the Lebesgue measure on $S^1$.
Its probability density function $\rho_f(x)$ is at least $C^1$  ( \cite{Ba} )
and depends on $f$ differentiably in $C^{r-1}$ topology. It yields the Fr\'echet differentiability of
  the entropy of $f$ with respect to the measure $\rho_f dx$, i.e.,
   $H(f) = \int_{S^1} \ln f'(x) \rho_f(x) dx .$  We want to prove  that $H(f)$ is also a Fr\'echet differentiable functional with
respect to the new Hilbert metric on $F_r$.  This can be done by considering the transfer operator over
the Sobolev space of density functions. We leave the proof in this general case to another article since
in this paper, we restrict our study to
a simpler case when the SRB measure $\rho_f$ is preserved by the perturbation of $f$.
When $\rho_f(x)$ is independent of $f$, the differentiability with respect to the Hilbert metric is much easier
to prove. We can directly calculate the derivative operator and prove the
Fr\'echet differentiability with respect to the Sobolev norm.

\subsection{Hilbert manifold of expanding maps preserving the Lebesgue measure}

We now define a Hilbert manifold $F_r(\rho),  r\ge 2$ to be the subset of $F_r$ consisting of maps
that preserve the same invariant measure with a density function $\rho(x)$.
 We may assume $\rho(x) =\rho_0(x)=1$ by changing the Riemannian metric
  on the circle \cite{DM, J21}. The corresponding subset is denoted by $F_r(\rho_0)$.

Given any map $f\in  F_r(\rho_0)$, the invariance of the Lebesgue measure under $f$ is characterized by the equation
\begin{equation}\label{eq:transf}
 1 = \rho_0(x)=\sum_{i=1}^n  \frac{\rho_0(y_i) }{f'(y_i) }=\sum_{i=1}^n  \frac{1}{f'(y_i) }, x \in [0,1)\end{equation} where $ 0\le y_1 < y_2< \cdots<  y_n < 1$
 are $n$ preimages of $x \in [0,1):$  $\tilde f(y_i) = x, \text{mod\ 1}, i=1,2, \cdots, n.$

We see that not only the equation (\ref{eq:transf}) is nonlinear in $f$, the points $\{y_i\}$,    preimages of $x$,
   also depend  on $f$. Thus, it is not convenient when we
calculate Gateaux derivatives of the entropy functional with respect to $f$. Instead,
we now identify the subset $F_r(\rho_0)$ with another Hilbert manifold with a Sobolev tangent space
 where the same SRB entropy functional's properties are much easier to study.

For each map $f \in F_r(\rho_0)$, we consider its lift $\tilde f$'s inverse map
$g(y), y \in [0, n]$. The $r$th derivative of $\tilde f$ is in $L^2[0,1]$ if and only if
$g(y)$'s $r$th derivative is in $L^2[0,n]$ since we have $\tilde f'(x) > 1, x \in [0,1]$.
$\tilde f^{(k)}(0)=  \tilde f^{(k)}(1), k=1,\dots r$ if and only if
 $g(y)$ is differentiable in $[0, n]$ up to order $r$ and $g^{(k)}(0^+) = g^{(k)}(n^-)$,
  $k=1,\dots r$. That is, $g(y)$ can be extended to a function whose
  $k$th derivative $1 \le k \le r$ is a period $n$ function. Given any $x \in [0, 1)$,
  $\tilde f$ maps   $n$ preimages of $x\in [0,1)$ under $f$ to $x, x+1, \cdots, x + (n-1)$
  in the universal covering space:  $\tilde f(y_i) = x+ i-1, i=1,2,\cdots, n$.  Thus,
   we have $$\frac{1}{f'(y_i) } = g'(x + i-1),\ i=1,2,\cdots, n $$ and the invariance of the Lebesgue
   measure becomes an equation linear in $g$:
 $$ 1 = \sum_{i=1}^n g'(x + i-1)     ,  x \in [0,1).$$

We now define the Hilbert manifold where we consider the SRB entropy's gradient for $r \ge 2$.
$$G_r=\{ g(y) \in C^{r-1}[0,n]: g(0)=0, g(n)=1, 0< g'(x) < 1,                g^{(r)} \in L^2[0, n], $$
 $$ \qquad g^{(k)}(0^+) = g^{(k)}(n^-), 1 \le k \le r, \sum_{i=1}^n g'(y + i-1)  =1, y \in [0,1)\}.$$

For each $g \in G_r$, its open neighborhood is identified with an open neighborhood in the Sobolev space, still denoted by $\Phi_r$:
\begin{align*}\Phi_r= &\{ \phi: \phi^{(r)} \in L^2[0, n], \phi(0)=0,  \phi^{(k)}(0^+) = \phi^{(k)}(n^-),
 0 \le k \le r,  \\
 &  \sum_{i=1}^n \phi'(y + i-1)  =0, y \in [0,1)\}. \end{align*}

The Sobolev norm on $\Phi_r$ is defined in the same way:
\[  \| \phi(x) \|^2_{H^r} =  \sum_{k=0}^r  \int_0^n [\phi^{(k)}]^2 (y) dy  .\]

For convenience, we denote the even larger family of functions without the constraint of preservation of Lebesgue measure by
$\bar  G_r$ and $\bar \Phi_r$:
$$\bar G_r=\{ g(y)  \in C^{r-1}[0,n], g(0)=0, g(n)=1, 0< g'(x) < 1,                g^{(r)} \in L^2[0, n], $$
$$ \qquad g^{(k)}(0^+) = g^{(k)}(n^-), 1 \le k \le r \},$$
$$\bar \Phi_r=\{ \phi: \phi^{(r)} \in L^2[0, n], \phi(0)=0,  \phi^{(k)}(0^+) = \phi^{(k)}(n^-), 0 \le k \le r \}.$$
Indeed, $G_r$ is a submanifold of  $\bar G_r$ and $\Phi_r$ a subspace of $\bar \Phi_r$.

On the Hilbert manifold  $\bar G_r$, there is a nature metric $d(g_1, g_2)$ that is consistent with the Sobolev norm in the tangent space of $\bar G_r$:
$$d^2(g_1, g_2)= \sum_{k=0}^r   \int_0^n [g_1^{(k)}- g_2^{(k)}   ]^2 (y) dy  .$$
 Indeed, $g_1-g_2 \in \bar \Phi_r$ for any $g_1, g_2 \in  \bar G_r $. Thus,
$$d(g_1, g_2) =\|g_1-g_2\|_{H^r}.$$

Let $a_i$ denote the preimage of $i$ under $\tilde f$ for $i=0, 1, \cdots, n-1$. The SRB entropy defined for every $f \in F_r(\rho_0)$ becomes
$$H(f) = \int_{S^1} \ln f'(x) d x =\sum_{i=1}^n  \int_{a_{i-1}}^{a_{i}}  \ln \tilde f'(x) d x $$
\begin{equation}\label{eq:entropy}
=\sum_{i=1}^n  \int_{ i-1}^{i} \ln \frac{1}{g'(y_i)}\  g'(y_i) d y_i = - \int_{0}^{n} \ln {g'(y )}\ g'(y ) d y=: H(g). \end{equation}

{\bf Remarks.} (1) We point out   an interesting connection between the SRB entropy of
 measure-preserving expanding maps and the Gibbs entropy of a probability
  measure with a density. Any function $g(y) \in \bar G_r$ can be considered as a probability measure on $[0, n]$
   with a density function $0<g'(y)<1$. $H(g)$ is then precisely the Gibbs entropy of a probability measure.

(2) We also point out  similarities and differences  between our approach and the approach  of Jordan, Kinderleherer, and Otto (JKO) \cite{JKO1, JKO2} in
their study of the gradient flow of the entropy functional (see also \cite{Maas}).  The main similarity is that both approaches start from the Gibbs
 entropy. But  two approaches have major differences:   In JKO's approach, the entropy  (or the Gibbs-Boltzmann entropy
  as it is called in  \cite{JKO1, JKO2}.  For  more discussions on Gibbs and Boltzmann entropy, see \cite{GL19} )  is defined for
   probability density functions on $\R^n$ not associated with any dynamical system. They use a discretized  process to obtain
    an approximate orbit from an initial density and then show  that  the orbit converges to an orbit from the heat equation as
    the step-size approaches zero. In our approach, the SRB entropy is defined for a chaotic dynamical system. That the entropy
    formula (\ref{eq:entropy}) taking the Gibbs entropy form seems to be coincidental. For Anosov systems on a higher dimensional
    manifold, the SRB entropy is defined for a higher dimensional map. The entropy formula may not take this particular form.
      The SRB entropy changes as the underlying chaotic dynamical system varies and we directly calculate the entropy functional's
      gradient under the Sobolev norm.   Our approach  leads to a system of countably many ordinary differential equations where the vector field is defined
       via integrals.  The system does have close connection to a nonlinear partial differential equation,  a gradient-dependent diffusion
       equation on the unit circle. See Section \ref{sec.numerical} for details.

We now state main results on the SRB entropy functional $H(g)$ on the Hilbert manifold $G_r$.

\begin{theorem}\label{thm1}

(1) The SRB entropy functional $H(g)$ is Fr\'echet differentiable on $G_r$.

(2) The gradient vector field of $H(g)$ is well-defined: for each $g \in G_r$, there exists a unique vector $X(g) \in \Phi_r$ such that
the directional derivative $<DH(g), X(g)/\|X\|_{H^r}>$ is the unique maximum among
all directional derivatives.

(3) The gradient vector field $X(g) \in G_r$ is Lipschitz continuous in $g$.

(4) There is a unique critical point for the gradient vector field $X(g)$ at the point where $g$ is a linear function on $[0, n]$.
\end{theorem}

An immediate consequence of  Theorem \ref{thm1} is that the differential equation defined on $G_r$ by
$ \frac{d \F_t}{d t }\big|_{t=0} = X(g)$ has a unique local solution for $t \in (-\epsilon, \epsilon)$: $\F_t(g) = \F(t, g) $ is a local flow defined on
$(-\epsilon_g, \epsilon_g) \times G_r $.

\section{Proof of Theorem \ref{thm1} }\label{sec.proof1}

{\bf Proof of Theorem \ref{thm1} (1)}

  We need to show that $H(g)$ is Fr\'echet differentiable. Since $G_r$ is a submanifold of $\bar G_r$.
  We can just prove  differentiability in $\bar G_r$. That means we do not need to consider the constraint of preservation of the Lebesgue measure.

We first calculate the first order term in $\epsilon$ of the difference
$H(g + \epsilon \phi) - H(g)$.

$$H(g+ \epsilon \phi ) = - \int_0^n \ln( g'(y) +\epsilon \phi'(y))  ( g'(y) + \epsilon \phi'(y)) dy .$$

Note that $$\ln( g'(y) + \epsilon \phi'(y)) = \ln  g'(y) (1 + \epsilon \phi'(y)/g'(y)   ) $$
$$= \ln g'(y)  + \ln (1 +\epsilon \phi'(y)/g'(y)   )$$
$$= \ln g'(y)  +  \epsilon \phi'(y)/g'(y) + O([\epsilon\phi'(y)/g'(y) ]^2) .$$
We have
$$ \ln( g'(y) + \epsilon \phi'(y))  ( g'(y) + \epsilon \phi'(y)) $$
$$= g'(y) \ln g'(y)  + \epsilon \phi'(y)+ \epsilon \phi'(y) \ln g'(y)  + O([\epsilon\phi'(y)]^2/g'(y)) + O( [\epsilon\phi'(y)]^3/[g'(y) ]^2).$$

Thus,  the first order term  of $H(g + \epsilon \phi) - H(g)$ in $\epsilon$
is
$$ - \int_0^n  ( 1+ \ln g'(y) ) \phi'(y) dy,$$ Since $\phi(0)=\phi(n) =0$, we have
the derivative operator   formula
\begin{align}\label{eq:dev}  DH_g\phi & = <DH(g), \phi> = - \int_0^n    \ln g'(y) \   \phi'(y) dy \\
& =  - \int_0^n    \ln g'(y) \    d \phi(y)  =    \int_0^n     \frac{g''(y)}{g'(y)} \   \phi(y) dy . \end{align}

We now show that $H(g)$ is Fr\'echet differentiable on $\bar G_r$, i.e, for any given $g \in \bar G_r$,
\[ \lim_{\epsilon \to 0}  \sup_{   \|\phi\|_{H^r} = 1 }
\frac{1}{\epsilon}[H(g + \epsilon \phi) - H(g)  - \epsilon DH_g \phi ] =0.\]
According to our earlier calculation,
\[ | H(g + \epsilon \phi) - H(g)   - \epsilon DH_g \phi|\]
\[= \epsilon^2  \int_0^n O([\phi'(y)]^2/g'(y)) + \epsilon  O( [\phi'(y)]^3/[g'(y) ]^2) dy  .\]

We now need a simple lemma on the upper bound of $|\phi'(y)|$:
\begin{lemma}\label{lm:ub}
Given any function $\phi \in \Phi_r$, $r\ge 2$,
$|\phi'(y)| \le M  \| \phi(y) \|_{H^2}$.  In general,
$| \phi^{(k)} | \le M \| \phi(y) \|_{H^{k+1}}$, where $1 \le k < r $ and $M$ is constant independent of $\phi$.
\end{lemma}

 {\it Proof of Lemma \ref{lm:ub}}

Since $\phi$ is of period $n$ and $\phi^{(r)} \in L^2[0,n]$, we have $\phi$'s Fourier expansion
\[  \phi(y) = a_0 + \sum_{k=1}^\infty  a_k \cos \frac{ 2 k \pi}{n} y  + b_k \sin \frac{ 2 k \pi}{n} y,\]
where the Fourier coefficients $a_k, b_k, k\ge 1$ satisfy the condition
$$   \sum_{k=1}^\infty  k^{2r} (a^2_k   + b^2_k ) < \infty .$$ For convenience,
we may assume $a_0=0.$
Thus,
$$| \phi(y) | = |   \sum_{k=1}^\infty  a_k \cos \frac{ 2 k \pi}{n} y  + b_k \sin \frac{ 2 k \pi}{n} y           |$$
$$ \le  \sum_{k=1}^\infty  |a_k | + |b_k| =    \sum_{k=1}^\infty \frac{1}{k} (k |a_k | + k |b_k|)
\le 2 \left[\sum_{k=1}^\infty  \frac{1}{k^2} \right]^{\frac{1}{2}} \left[ \sum_{k=1}^\infty  k^{2} (a^2_k   + b^2_k )\right]^{\frac{1}{2}} .$$
Let $M= 2 \left[\sum_{k=1}^\infty  \frac{1}{k^2} \right]^{\frac{1}{2}}$ and
 apply the inequality above to $\phi'(y)$. We have
 $$ |\phi'(y)| \le M \|\phi(y)\|_{H_2}.$$

$\hfill \Box$

Since $g'(y)>0$ is bounded away from $0$ and $\phi^{(2)} \in L^2[0,n]$, the integrand
$ O([\phi'(y)]^2/g'(y)) + \epsilon  O( [\phi'(y)]^3/[g'(y) ]^2)$ is a bounded function over $[0,n]$.
 We have
$$\lim_{\epsilon\to 0} \epsilon \int_0^n O([\phi'(y)]^2/g'(y)) + \epsilon  O( [\phi'(y)]^3/[g'(y) ]^2) dy=0.$$ $H(g)$ is Fr\'echet differentiable at any $ g \in G_r$.

{\bf Proof of Theorem \ref{thm1}  (2) \& (4)}

We pick an orthonormal basis of $\Phi_r$: $\{\textbf{e}_i\}_{i=1}^\infty$.
Define  $$X(g)=   DH_g= \sum_{i=1}^\infty b_i \textbf{e}_i, $$ where  $b_i=   DH_g (\textbf{e}_i) , i=1,2, \cdots.$

Given any $\phi \in \Phi_r$ with $\| \phi\|_{H^r} =1$ and $\phi= \sum_{i=1}^\infty <\phi, \textbf{e}_i>_{H^r} \textbf{e}_i $, where $<\cdot>_{H^r}$ denotes the inner product of $\Phi_r$, we have
$$DH_g(\phi) =\sum_{i=1}^\infty b_i <\phi, e_i>_{H^r}  .$$
$|    DH_g(\phi)        |$ reaches maximum if and only if $ <\phi, e_i>_{H^r} =C b_i, i = 1,2, \cdots$ for some constant $C$, i.e.
$\phi = C X(g)$. Since $\| \phi\|_{H^r} =1  $, we have $C = [ \| X(g)\|_{H^r}             ]^{-1}.$

Thus,  the gradient vector field $X(g)$ is well-defined on Hilbert manifold $G_r$. Since we know that  $X(g)=0$ if and only if
$g$ is linear \cite{J21}, $X(g) \not=0$ for all $g \in G_r$ except for $g$ linear.

{\bf Proof of Theorem \ref{thm1} (3)}

We now prove that the derivative operator $DH_g$  is Lipschitz, which leads to the local existence of the gradient flow \cite{LY}.

Given any two maps $g_1, g_2 \in G_r$,   denote $\psi=    g_2 - g_1 \in \Phi_r$.
We now estimate the distance between two derivative operators $DH_{g_1}$ and $DH_{g_1+\psi}$.
When $g_1$ and $g_2$ are close, their open neighborhoods overlap.
Thus, we can assume two derivative operators $DH_{g_1}$ and $DH_{g_2}$ are acting on the same Sobolev space $\Phi_r$.
 For simplicity of notation, we drop the subscript in $g_1$.  We have
$$ \| DH_{g + \psi} -  DH_{g } \| = \sup_{\|\phi\|_{H^r}=1} | \int^n_0 [ \ln ( g' + \psi') - \ln(g')] \phi' dy|$$
$$ \le \sup_{\|\phi\|_{H^r}=1}  \left[\int_0^n | \ln ( g' + \psi') - \ln(g')    |^2 d y   \int_0^n |\phi'|^2 d y              \right]^{1/2}$$
$$ \le  \left[\int_0^n | \ln ( g' + \psi') - \ln(g')    |^2 d y \right]^{1/2}   =     \left[\int_0^n    | \ln ( 1 + \frac{\psi'}{g'}) |^2 d y \right]^{1/2}      $$

Since $g' > 0$ is bounded from below and $\psi \in \Phi_r, r\ge 2$,
we may assume that $|\psi'|$ is sufficiently small and $|\frac{\psi'}{g'}| < \delta< 1$.  So, there is a constant $K$ such that $$| \ln ( 1 + \frac{\psi'}{g'}) | \le K |\frac{\psi'}{g'}|.$$

Thus, we have
$$ \| DH_{g + \psi} -  DH_{g } \| \le \left[\int_0^n   K |   \frac{\psi'}{g'}  |^2 d y \right]^{1/2}   \le C \| \psi \|_{H^r} ,   $$ where
$C = [\frac{K}{\displaystyle \min_{y \in [0,n]} g'(y) }]^{1/2}.$ We conclude that $DH_g$ is Lipschitz continuous over $G_r$.

{\bf Remark}  Since $\ln x$ is an analytic function in $x$ in a  small neighborhood of
any $x_0 > 0$, we can in fact show that  $H(g)$ is analytic in $g$ on Hilbert manifold $G_r$.

\section{Global Existence of the Gradient Flow}\label{sec.global}

We now study the global existence of the gradient flow of the SRB entropy for $t \in [0, \infty)$
and prove that every trajectory converges to the unique equilibrium where the expansion rate of the map is a constant.
 We note that under the Sobolev norm, the gradient vector at point $g \in G_r$ is defined by the integral  $ \int_0^n     \frac{g''(y)}{g'(y)} \   \phi(y) dy$.
 While the integral does define a linear functional in the tangent space $\Phi_r$ for each $ g \in G_r$. Its Riesz representation in $\Phi_r$ is, in general, not $ \frac{g''(y)}{g'(y)}$ since it is in general not a vector in $\Phi$. This poses an obstacle for proving the global existence, even though,
  the global existence is likely true. We instead reconsider the global existence of the gradient flow in a different Hilbert metric on a slightly different Hilbert manifold.

We first expand the domain of the entropy functional to a larger Hilbert space where the gradient vector's  Riesz representation can be obtained explicitly.

For the Hilbert manifold $$ G_r=\{  g(y) \in C^{r-1}[0,n]:  y\in [0,n];  g(0)=0, g(n)=1; 0 < g'(y) < 1; g^{(r)} \in L^2[0,n];$$
$$\qquad \qquad g^{(k)}(o^+) = g^{(k)}(n^-), 1 \le k \le r-1, r\ge 2, \sum_{i=0}^{n-1} g'(y +i)  =1, y \in [0,1] \},$$
each map $g(y)\in G_r $ is uniquely defined by its derivative:  $g(y) = \int_0^y  g'(\tau) d \tau $. So,  we can embed
 $ G_r$ into another Hilbert manifold $G'$:
$$G' =\{ h(y)  \in L^2[0, n]:  0<  h(y) < 1\  (a.e.),  \int_0^n h(y) dy = 1,  \sum_{i=0}^{n-1} h(y+i)  =1, y \in [0,1] \} .$$
Given any $g(y) \in   G_r$, the embedding map is defined  by  $$g(y) \to h(y) = g'(y) \in G'.$$  $G'$ is clearly a Hilbert manifold
with a tangent space
\begin{equation}\label{eq:tangent}
\Psi =\{ \psi(y)  \in L^2[0, n],    \int_0^n \psi(y) dy = 0,   \sum_{i=0}^{n-1} \psi(y+i) =0, y \in [0,1] \}, \end{equation}
equipped with a common Hilbert norm $\| \phi \| = \int_0^n \psi^2(y) dy .$

Note that the condition  $ \int_0^n \psi(y) dy = 0 $ in (\ref{eq:tangent}) can be removed since it is implied by the condition $ \sum_{i=0}^{n-1} \psi(y+i)     =0$:
$$  \int_0^n \psi(y) dy = \int_0^1 \psi(y) dy + \int_1^2 \psi(y) dy + \cdots + \int_{n-1}^n \psi(y) dy$$
$$= \int_0^1 \psi(y) dy +  \int_0^1 \psi(z+1) dz +\cdots + \int_0^1 \psi(z+n-1) dz =   \int_0^1 \sum_{i=0}^{n-1} \psi(y+i) dy=0.$$
The tangent space $\Psi$ is a subspace of $ L^2[0, n]$.

Note also that
the entropy functional, $$H(g)= - \int_0^n \ln g'(y)\  g'(y) dy$$ on $ G_r$ becomes
 $$H(h) = - \int_0^n \ln h(y)\  h(y) dy,$$ which is well defined on entire $G'$ and the Gateaux
  derivative of $H(h)$ in the direction of $\psi \in \Psi$ exists and
has the same formula (see (\ref{eq:dev})):
$$DH_h(\psi) = - \int_0^n \ln h(y)\  \psi dy.$$ This Gateaux derivative defines a bounded linear functional on the tangent space $\Psi$.

A direct calculation will confirm  that this linear functional's Riesz representation is given by
$$R_h(y) =-    \ln h(y) + \frac{1}{n}\sum_{i=0}^{n-1} \ln h(y+i)  ,$$ where $h(y)$ is extended into a period $n$ function over $[0, \infty)$.
Indeed, we can easily verify that $R_h(y) \in \Psi$.    We only need to verify that $ \int_0^n R_h(y) \psi dy = - \int_0^n \ln h(y)  \psi dy$ for all $\psi \in \Psi$ since the identity
$$ \sum_{i=0}^{n-1} R_h(y+i)  =0$$ for all $y \in [0, 1]$   clearly holds  due to the periodicity of $h(y)$.

To see that $ \int_0^n R_h(y) \psi dy = - \int_0^n \ln h(y)  \psi dy$ for all $\psi \in \Psi$, we first extend $\psi$ to a period $n$ function and
calculate the following integral applying   integration by substitution and periodicity of both functions $h(y)$ and $\psi(y)$:
$$ \int_0^n  \sum_{i=0}^{n-1} \ln h(y+i)  \ \psi(y)  dy .$$ Let $z= y+i$ in each integral,
we have
$$ \int_0^n  \sum_{i=0}^{n-1} \ln h(y+i)  \ \psi(y)  dy = \int_0^n \ln h(z) \ \sum_{i=0}^{n-1}  \psi(z-i)  dz=0.$$

We summarize the properties of the entropy functional  $H(h)=- \int_0^n \ln h(y)\  h(y) dy $ over the Hilbert manifold $G'$ in the following proposition.
\begin{proposition}\label{prop:G'}
(1)  $H(h)$ is Gateaux differentiable at every $h \in G'$ and the derivative formula in the direction of $\psi \in \Psi$ is given by a continuous linear functional on $\Psi$:  $$DH_h(\psi) = - \int_0^n \ln h(y)  \psi dy .$$

(2) The Riesz representation of the derivative operator $DH_h $ over $\Psi$  is $$R_h(y)= - \ln h(y) + \frac{1}{n}   \sum_{i=0}^{n-1}\ln h(y+i) ,$$ where $h(y)$ is extended periodically to $[0, \infty)$.

(3) The maximum value of  $DH_h(\psi)$ over $\psi \in \Psi$ with $   \int_0^n \psi^2(y) d y = 1 $
is reached at the unit vector  $ R_h(y) [ \int_0^n R_h^2(y) dy]^{-\frac{1}{2}} .$

\end{proposition}

We denote this gradient vector field over $G'$ by $$Y(h)=R_h(y)= - \ln h(y) + \frac{1}{n}   \sum_{i=0}^{n-1}\ln h(y+i).$$
 It is Lipschitz continuous in terms of $h$ under the $L^2$ norm, thus, locally integrable.

We now prove the following theorem on the global existence of the gradient flow of the SRB entropy
and the convergence of every flow trajectory to a unique equilibrium as $ t\to \infty.$

\begin{theorem}\label{thm.global} For Lebesgue measure preserving $C^{1+\alpha}$ expanding maps on the circle,
the SRB entropy functional $H(f) = \int_0^1 \ln f'(x) dx $ induces a gradient flow on the space of derivatives of inverse maps under the $L^2$ norm.
This gradient flow exists globally for all $t \in [0, \infty)$ and every trajectory converges to the unique equilibrium corresponding to the linear expanding map.
 \end{theorem}

\begin{proof}

For any fixed initial map $h(y) \in G'$, let ${\mathcal G}_t(h)=g(t, h)$ denote the local flow defined
by the gradient vector field $Y(h)$ on $G'$ for $t \in (- \epsilon_h, \epsilon_h)$.
For any $y \in [0,1]$, We have
$$ \frac{d}{d t} g(t, h(y) )|_{t=0} = - \ln h(y) + \frac{1}{n}   \sum_{i=0}^{n-1}\ln h(y+i).$$
By periodicity of $h(y)$, for all $k=1,2, \cdots, n-1$,  we also have
$$ \frac{d}{d t} g(t, h(y+k) )|_{t=0}  = - \ln h(y+k) + \frac{1}{n}   \sum_{i=0}^{n-1}\ln h(y+i).$$

Introduce new variables $x_k = h(y+k-1)$, $ k=1,2,\cdots, n$, we have a system of $n$ ordinary differential equations
$$
{\dot x}_k = -   \ln x_k  + \frac{1}{n} \sum_{i=1}^{n}  \ln x_i   , k=1,2, \cdots, n
 $$
subject to the condition $  \sum_{k=1}^n x_k =1$. Consequently, we have
$$\begin{cases}
{\dot x}_1 - {\dot x}_2 &= -   (  \ln x_1 - \ln x_2   )  \\
{\dot x}_2 - {\dot x}_3 &= -   (  \ln x_2  - \ln x_3 )  \\
\quad \cdots \cdots &=\qquad  \cdots \cdots\\
{\dot x}_n - {\dot x}_1 &= -   (  \ln x_n   - \ln x_1 ).
 \end{cases}$$

The solution to the system exists globally for all initial values in the region $0< x_1, x_2, \cdots,  x_k < 1 $
on the invariant plane $ x_1+x_2 + \cdots + x_n =1$
and all solutions converge to the unique equilibrium $x_1=x_2 =\cdots = x_n=\frac{1}{n}$.
\end{proof}

\section{Differential equation representation of the gradient flow}\label{sec.numerical}

We now explore  the possibility of representing the gradient flow $\mathcal F_t(g)$ from Section \ref{sec.local} as explicit differential equations.

Let ${\mathcal F}_t$ denote the gradient flow defined by the vector field $X(g)$ over $G_r$,i.e,
 $  {\mathcal F}_t(g) $ is a map from $ (- \epsilon, \epsilon) \times G_r\to G_r$ differentiable in $t$ and ${\mathcal F}_0(g)=g$
 and
 $$  \frac{d}{dt} {\mathcal F}_t(g)\big|_{t=0} = X(g),$$ where $X(g) \in \Phi_r$ is defined by an integral operator
 $$<X(g), \phi> = - \int_{0}^n \ln g'(y) \ \phi'  dy = \int_{0}^n   \frac{g''(y)}{g'(y)}  \phi(y)
 dy  .$$

 We  see that maps in $ \bar G_r$ can be easily
 represented as a series.  In the simple case when $n=2$,
 we can obtain a system of ordinary differential equations that generates the flow.
 Numerical methods such as Euler's method \cite{Gi} can then be used to obtain   typical approximate trajectories of the flow.

 For any given $ g \in \bar G_r$, $g(y) - \frac{y}{n} \in \bar \Phi_r $ is a continuous periodic
  function of period $n$ and its derivative is bounded.
  Thus, its Fourier series converges to itself both pointwise and
  in the Sobolev norm. Thus, Hilbert manifold $\bar G_r$ can be
  represented as
 $$\bar G_r =\{ \frac{y}{n} + \sum_{k=1}^\infty a_k \cos \frac{2\pi k }{n} y
 + b_k \sin \frac{2\pi k }{n} y, y \in [0,n],  \sum_{k=1}^\infty k^{2r} (a^2_k +b^2_k) < \infty\}, $$
 where $a_k, b_k$ are Fourier coefficients of $g(y) - \frac{y}{n}$ satisfying the condition
 $$ 0< \frac{1}{n}  + \sum_{k=1}^\infty -  \frac{2a_k\pi k }{n} \sin \frac{2\pi k }{n} y
  +  \frac{2b_k\pi k }{n} \cos \frac{2\pi k }{n} y  < 1.$$  Notice that we have replaced the condition $g(0)=0$ by
  dropping the constant term in the Fourier series since the entropy is a function of $g'(x)$.
  This adjustment is also made to the tangent space $\bar \Phi_r$.

 Maps in the submanifold $G_r$ will have to satisfy an addition linear equation:
  \begin{equation} \label{eq:linear}  \sum_{i=1}^n  \sum_{k=1}^\infty
   -  k a_k \sin \frac{2\pi k }{n} (y+i-1)
   +  kb_k \cos \frac{2\pi k }{n} (y +i -1)=0, y \in [0,1].  \end{equation}

Assume that we have an orthonormal  basis of $\Phi_r$: $\{ \textbf{e}_k \}_{k=1}^\infty$.  Then,
 any trajectory of the flow $ u(t,y)= {\mathcal F}_t(g) $ can be written in the form
$$ u(t,y) =  \frac{y}{n} + \sum_{i=1}^\infty c_k(t) {\bf e}_k,$$ with
  $u(y, 0)=\sum_{k=1}^\infty c_k(0) {\bf e}_k = g -\frac{y}{n} $.
  Thus the flow equation $\frac{d \F_t(g)}{dt}\big|_{t=0} = X_g$ becomes
  $$   \sum_{k=1}^\infty \dot c_k(t)\ {\bf e}_k=  \sum_{k}^\infty <X_{u(t,y)}, \textbf{e}_k> \textbf{e}_k. $$
  We have a system of countably many ordinary differential equations:
  \begin{equation}\label{eq:general}\dot c_k(t) = \sum_{i}^\infty <X_{u(t,y)}, \textbf{e}_k>
   =\int_0^n \frac{u_{yy}}{u_y}\  \textbf{e}_k dy, \end{equation}
   where $$u_y = \frac{\partial u(t,y)}{\partial y}
   =\frac{1}{n}+ \sum_{k=1}^\infty c_i(t) \frac{d {\bf e}_k}{dy},\ \text{and}$$
  $$u_{yy} = \frac{\partial^2 u(t,y)}{\partial y^2}
   =  \sum_{k=1}^\infty c_i(t) \frac{d^2 {\bf e}_k}{dy^2}.$$

While it is easy to obtain a set of orthonormal basis
for $\bar \Phi_r$ since the set $\{ \cos \frac{2\pi k }{n} y      ,
   \sin \frac{2\pi k }{n} y     \}_{k=1}^\infty $ is clearly
    an orthogonal basis, the linear constraint (\ref{eq:linear}) poses
      an obstacle to finding orthogonal basis for $\Phi_r$.
      Fortunately, in the simple case when $n=2$, an orthonormal basis for
      $\Phi_r$ can be obtained directly from this set. That will allow us to
       obtain a system of countably many  ordinary differential equations
        explicitly and thus, to  approximate numerically typical trajectories of
         the flow.

 \subsection{ Ordinary differential equation representation  when $n=2$}

 We now look at the case when $n=2$. $r \ge 2$ can be any number.  In this case,
 the linear constraint (\ref{eq:linear}) becomes
 $$  \sum_{k=1}^\infty  {k a_k   }(1+ (-1)^k)  \sin  { \pi k y}
  +   { k b_k    } (1+ (-1)^k)  \cos  { \pi k y }   =0 .$$
  So, we can conclude that  $a_k=b_k =0$ when $k$ is even. For simplicity, we also let $r=2$.
 Since $$\int_0^2  \cos^2 {\pi k }  y   dy =  \int_0^2   \sin^2  { \pi k }  y dy = 1,$$ we have
 \[\| \cos  \pi k   y \|_{H^2} =\| \sin  { \pi k }  y \|_{H^2}
 = \left[  (1+   (k\pi)^2 + (k\pi)^4  )\right]^{1/2} =:\frac{1}{c_k}.   \]
 So, the set $\{ c_{2m-1}  \cos (2m-1) \pi y,  c_{2m-1}  \sin (2m-1) \pi y\}_{m=1}^\infty$
 is an orthonormal basis of $\Phi_2$.

 Let $u(t,y)$ be a trajectory of the flow $\F_t(g)$. For each $t$,
 $$u(t,y) =  \frac{y}{2}  + \sum_{k=1}^\infty    { a_{2k-1}(t)   }  \cos [ (2k-1) \pi   y]
  +  b_{2k-1} (t)    \sin [(2k-1)\pi y]  \in G_2.$$  For a fixed value of $t$, the gradient vector at $u(x,t)$ is
  $$X(u(t,y)) = \sum_{m=1}^\infty <X(u(t,y)), c_{2m-1} \cos (2m-1)\pi  y > c_{2m-1} \cos (2m-1)\pi  y $$
    $$+<X(u(t,y)), c_{2m-1} \sin (2m-1)\pi  y > c_{2m-1} \sin (2m-1)\pi  y       $$
   $$=\sum_{m=1}^\infty c_{2m-1}^2 \left[ <X(u(t,y)),   \cos (2m-1)\pi  y >  \cos (2m-1)\pi  y \right. $$
   $$ \qquad \qquad \qquad \left.  +<X(u(t,y)),   \sin (2m-1)\pi  y >  \sin (2m-1)\pi  y  \right]$$
   $$=\sum_{m=1}^\infty c_{2m-1}^2 \left[ \int_0^2 \frac{u_{yy}}{u_y}  \cos   (2m-1) \pi  y  dy\   \cos (2m-1)\pi  y \right. $$
   $$ \qquad \qquad \qquad   +\left.\int_0^2 \frac{u_{yy}}{u_y}  \sin   (2m-1) \pi  y  dy \ \sin (2m-1)\pi  y  \right].$$

   Notice that $$u_t=\frac{d \F_t(g)}{dt} = \sum_{m=1}^\infty    { \dot a_{2m-1}   }  \cos [ (2m-1) \pi   y]
  +  \dot b_{2m-1}     \sin [(2m-1)\pi y]  \in \Phi_2,$$
   We obtain explicitly a system of ordinary differential equations defined on $G_2$ that
    generates the gradient flow.i.e., the local flow $\mathcal F_t(g)$ is the
    solution to the system of differential equations:
 \begin{equation}\label{eq:ode1}
  \dot a_{2m-1} = c_{2m-1}^2  \int_0^2\frac{u_{yy}}{u_y}  \cos   (2m-1) \pi  y  dy; \end{equation}
 \begin{equation}\label{eq:ode2}
  \dot b_{2m-1} = c_{2m-1}^2   \int_0^2 \frac{u_{yy}}{u_y}  \sin  (2m-1) \pi   y dy, \end{equation}
 where $$u_y =  \frac{1}{2}  + \pi \sum_{k=1}^\infty    (2k-1)[-   { a_{2k-1}   }  \sin [ (2k-1) \pi   y]
  +  b_{2k-1}      \cos [(2k-1)\pi y].$$

\subsection{The partial differential equation connection}

Since the gradient vector $X$ at $u(t,y)\in G_2$ is  defined by
$$<X( u(t,y)), \varphi> = \int_0^2  \frac{u_{yy}}{u_y} \varphi  dy,\ \varphi \in \Phi_2,$$
there is a close connection between the gradient flow $\F_t(g)$ and the solution to
the nonlinear partial differential equation $w_t = \frac{w_{yy}}{w_y}$, a gradient-dependent diffusion equation
defined on the unit circle.
Assume that
$$w(t,y) = \frac{y}{2} + \sum_{k=1}^\infty a_{2k-1}(t) \cos (2k-1)\pi y + b_{2k-1}(t) \sin (2k-1)\pi y \in G_2 $$
 is a solution to $w_t = \frac{w_{yy}}{w_y}$ in some open interval of $t$ and
$w_t $ and $    \frac{w_{yy}}{w_y}$ are both in $L^2[0, 2]$ for each $t$.
  We have $$\sum_{k=1}^\infty \dot a_{2k-1}(t) \cos (2k-1)\pi y + \dot b_{2k-1}(t) \sin (2k-1)\pi y  =
  \frac{w_{yy}}{w_y}$$ as functions of $L^2[0,2]$ for each $t$.
  By orthogonality of the set $$\{\cos (2m-1)\pi y, \sin (2m-1)\pi y \}_{m=1}^\infty$$  and
   $$\int_0^2 \cos^2 (2m-1)\pi y dy = \int_0^2 \sin^2 (2m-1)\pi y dy =1,$$ we have
 \begin{align}\label{eq:pde} \dot a_{2m-1}(t)  &=  \int^2_0  \frac{w_{yy}}{w_y} \cos(2m-1)\pi y dy , \\
    \dot b_{2m-1}(t)  &=  \int^2_0  \frac{w_{yy}}{w_y} \sin(2m-1)\pi y dy ,\end{align}
    where $$w_y =  \frac{1}{2} + \pi \sum_{k=1}^\infty (2k-1)[-a_{2k-1}(t) \sin (2k-1)\pi y + b_{2k-1}(t) \cos (2k-1)\pi y].$$

    The systems in (\ref{eq:pde}) and (\ref{eq:ode1}) differ only by a constant coefficient in front of each equation.

\subsection{Numerical approximation of a flow trajectory}

We limit the scope of our numerical exploration to the case when $g'$ is an even function:
$$g'= \frac{1}{2} + \pi \sum_{k=1}^\infty (2k-1)   b_{2k-1}\cos(2k-1)\pi y, \  g''=-  \pi^2  \sum_{k=1}^\infty (2k-1)^2  b_{2k-1}\sin(2k-1)\pi y .$$
We have $\dot a_{2m-1}= 0$ for all $m \in \mathbb N$. Thus, the system of ODEs in (\ref{eq:ode2}) is reduced to
\begin{equation}\label{eq:bm}
 \dot b_{2m-1} =  c_{2m-1}^2  \int_0^2 \frac{g''}{g'}  \sin  (2m-1) \pi   y dy, \end{equation}
where $$g' =  \frac{1}{2}  + \pi \sum_{k=1}^\infty  (2k-1)   b_{2k-1}   \cos[(2k-1)\pi y].$$
Let $\pi y=\tau. $  We have
$$ \dot b_{2m-1} =  c_{2m-1}^2  \int_0^{2\pi}  \frac{g''}{g'}     \sin [ (2m-1) \tau]  \frac{d\tau}{\pi},$$
where $g' =  \frac{1}{2}  + \pi \sum_{k=1}^\infty   b_{2k-1} (2k-1)   \cos[(2k-1)\tau$  and
$g''=-  \pi^2  \sum_{k=1}^\infty (2k-1)^2  b_{2k-1}\sin[(2k-1)\tau]$.
Denote $ B_k=\pi b_{2k-1}(2k-1)$,
$h(\tau) =  g' = \frac{1}{2}  +   \sum_{k=1}^\infty   B_k     \cos[ (2k-1)\tau]$ and
$h'(\tau) =   -   \sum_{k=1}^\infty   B_k  (2k-1)   \cos[ (2k-1)\tau]$.
So, the system (\ref{eq:bm}) becomes
\begin{equation}\label{eq:bmt}
 \dot b_{2m-1} =  c_{2m-1}^2  \int_0^{2\pi} \frac{h'}{h}  [\sin  (2m-1) \tau] d\tau. \end{equation}

 Replacing $ \dot b_{2m-1}$ in (\ref{eq:bmt}) by  $\frac{\dot B_m}{\pi (2m-1)}$, we have
$$ \dot B_{m} =  \pi (2m-1)   c_{2m-1}^2   \int_0^{2\pi} \frac{h'}{h}        \sin [ (2m-1)  \tau] d\tau,$$
where $$h(\tau) =  \frac{1}{2}  +   \sum_{k=1}^\infty   B_k     \cos[(2k-1)\tau].$$

Or, in one big formula,
\begin{align}\label{eq:big} \dot B_{m} &=    \frac{- \pi (2m-1)}{  1 +  (2m-1)^2 \pi^2 + (2m-1)^4\pi^4} \\
&\cdot  \int_0^{2\pi}  \frac{   \sum_{k=1}^\infty   B_k  (2k-1)   \sin[(2k-1)\tau]}{  \frac{1}{2}
 +   \sum_{k=1}^\infty   B_k     \cos[(2k-1)\tau]            }                \sin[ (2m-1) \tau]     d\tau .\end{align}

Let  $F_m( \{B_k\})$ denote the right hand side of the equation (\ref{eq:big}).

We use  Euler's method to approximate a trajectory of this system of ODEs:
$$  B_m( k \epsilon) = B_m(    (k-1) \epsilon ) + \epsilon F ( \{B_m( (k-1)\epsilon)\}),  k=1,2,\cdots,$$
where $\epsilon>0$ is the step size and the initial point is $\{B_m(0)\}$.

We choose $B_1(0)= \frac{1}{4}$ and $B_m(0)=0, m>1$, i.e., the initial map's derivative is $h(y)= \frac{1}{2} + \frac{1}{4} \cos y.$ We have
$$  B_m( \epsilon) = B_m( 0) + \epsilon F_m( \{B_m(0)\}), m\ge 0$$
where
$$F_m( \{B_m(0)\}) = -   \pi (2m-1)c_{2m-1}^2
 \int_0^{2\pi}   \frac{       \frac{1}{4}  \sin\tau}{  \frac{1}{2}
 +    \frac{1}{4}     \cos\tau            }          \sin(2m-1)\tau d \tau$$
$$= -   \pi (2m-1)  c_{2m-1}^2  \int_0^{2\pi} \frac{          \sin\tau}{  2
 +        \cos\tau            }    \sin(2m-1)\tau d \tau.$$
These values can be easily computed using numeric integration.
We see that $B_m(\epsilon)$ is generally not zero  for all $m \ge 1$,
regardless how small the step size $\epsilon>0$ is.

The numerical simulation of solutions to the system (\ref{eq:big})  is carried out on Maple by Maplesoft.
Since $B_m(k\epsilon)$ decays very fast in $m$, we have only kept three terms in the Galerkin method. The step size $\epsilon=0.1$ in Euler's method \cite{Gi}.

In Figure 1, graphes of $\frac{d}{dy} \F_t(g) - \frac{1}{2} $ are shown for three values of $t$, $t=0, $ $t= 10$, and $t= 20$:
\begin{align*} t=0, \quad &h_0 =  \frac{1}{2} +  \frac{1}{4} \cos \tau.\\
 t=10,  \quad   &h_{200} \approx  \frac{1}{2} +  0.121 \cos \tau - 0.000198 \cos 3 \tau - 0.00000196 \cos 5 \tau. \\
 t=20, \quad &h_{200} \approx  \frac{1}{2} +  0.043 \cos \tau - 0.000196 \cos 3 \tau - 0.00000186 \cos 5 \tau.
 \end{align*}

In Figure 2, we show the differences between   $c_m \cos(y)$ (from the trajectory of the heat equation $u_t =u_{xx}$ with the same initial value) and  $\frac{d}{dy} \F_t(g) - \frac{1}{2} $ for an even larger  $t$. The vertical axis is re-scaled by a factor of $1000$.
$t=50 ,$
$$  h_{500} \approx  \frac{1}{2} +  0.000431 \cos \tau - 0.000113 \cos 3 \tau - 0.00000152 \cos 5 \tau. $$

\begin{figure}\label{fig1}
  \caption{Graphs of the deviation  of the derivative of the inverse map from its equilibrium for various values of $t$ along a trajectory of the gradient flow.}
    \includegraphics[width=0.6\textwidth]{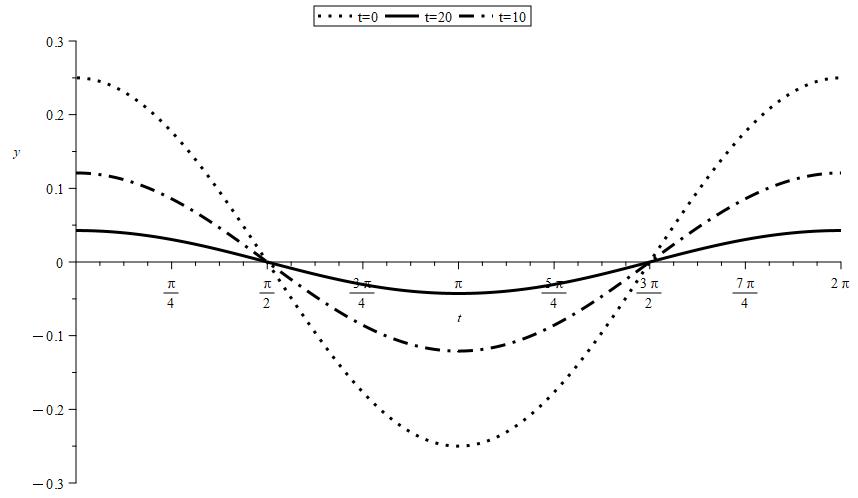}
\end{figure}

\begin{figure}\label{fig2}
  \caption{The dotted line is the graph of a cosine curve. The solid line is the graph of the deviation  of the derivative of the inverse map from its equilibrium when $t$ is large. The verical axis is re-scaled with a factor of 1000.}
     \includegraphics[width=0.6\textwidth]{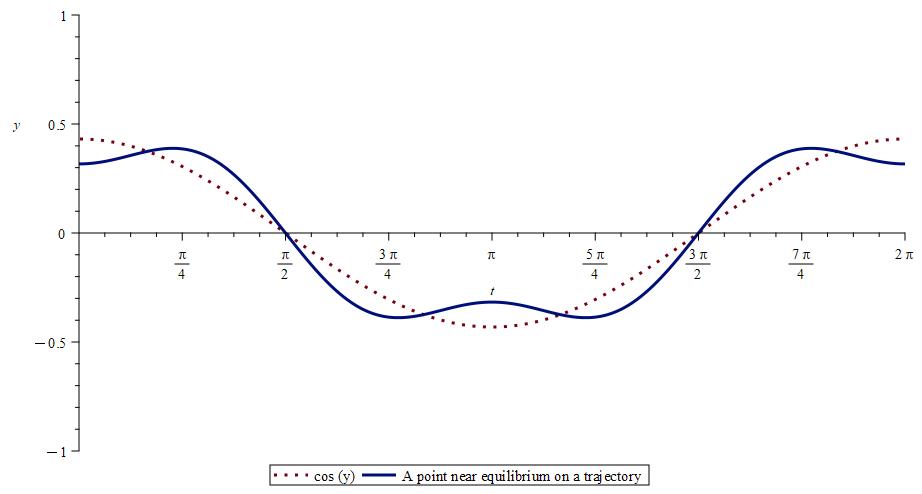}
\end{figure}

{\bf Ending Remarks}

 We see that the diffusion process from this gradient flow is different from that of the heat equation.
Due to the linearity, the flow from the heat equation does not create higher frequency terms  if the initial heat distribution does not have them.  In the gradient flow induced by the SRB entropy, the higher frequency terms appear immediately when $t$ increases even though the amplitudes of these high frequency terms are very small.

  Numerical evidence suggests that the gradient flow $\F_t(g)$ is also defined globally over $G_2$ and
$\lim_{t\to \infty} \F_t(g) = \frac{y}{2}$.  However, a rigorous proof is not available at the moment.

\section{Acknowledgement}
The author thanks John Gemmer, Sarah Raynor, Yang Yun and Yunping Jiang for many beneficial discussions.

Acknowledgement

\end{document}